\documentclass[review]{elsarticle}

\usepackage{lineno,hyperref}
\modulolinenumbers[5]

\journal{Discrete Applied Mathematics}

\usepackage{epsfig} 
\usepackage{amsthm}
\usepackage{color}
\usepackage[]{algorithm2e}
\usepackage{stackengine}
\usepackage{dsfont}
\usepackage{bbm}
\usepackage{amsmath,amssymb,amsthm}
\usepackage{comment}

\usepackage{algcompatible}

\newtheoremstyle{definition}
{9pt}           
{9pt}           
{}              
{0cm}           
{\bfseries}     
{\ }            
{ }             
{}              

\theoremstyle{definition}
\newtheorem{definition}{Definition}[section]
\newtheorem{theorem}{Theorem}[section]

\newtheorem{example}{Example}[section]

\begin{document}

\begin{frontmatter}	
\title{A Polynomial-Time Algorithm for Special Cases of the Unbounded Subset-Sum Problem}

\author{Majid Salimi}
\ead{M.Salimi@eng.ui.ac.ir}
\author{Hamid Mala\corref{cor1}}
 \ead{h.mala@eng.ui.ac.ir}

\address{Faculty of Computer Engineering, University of Isfahan, Isfahan, Iran}
\cortext[cor1]{Corresponding author}

\begin{abstract}
	The Unbounded Subset-Sum Problem (USSP) is defined as: given sum $s$ and a set of integers $W\leftarrow \{p_1,\dots,p_n\}$ output a set of non-negative integers $\{y_1,\dots,y_n\}$ such that $p_1y_1+\dots+p_ny_n=s$. The USSP is an NP-complete problem that does not have any known polynomial-time solution. There is a pseudo-polynomial algorithm for the USSP problem with  $O((p_{1})^{2}+n)$ time complexity and $O(p_{1})$ memory complexity, where $p_{1}$ is the smallest element of $W$ \cite{PH}. This algorithm is polynomial in term of the number of inputs, but exponential in the size of $p_1$. Therefore, this solution is impractical for the large-scale problems.\\
	In this paper, first we propose an efficient polynomial-time algorithm with $O(n)$ computational complexity for solving the specific case of the USSP where $ s> \sum_{i=1}^{k-1}q_iq_{i+1}-q_i-q_{i+1}$, $q_i$'s are the elements of a small subset of $W$ in which $gcd$ of its elements divides $s$ and $2\le k \le n$. Second, we present another algorithm for smaller values of $s$ with $O(n^2)$ computational complexity that finds the answer for some inputs with a probability between $0.5$ to $1$. Its success probability is directly related to the number of subsets of $W$ in which $gcd$ of their elements divides $s$. This algorithm can solve the USSP problem with large inputs in the polynomial-time, no matter how big inputs are, but, in some special cases where $s$ is small, it cannot find the answer.	
\end{abstract}

\end{frontmatter}
\section{Introduction} 
 The Subset-Sum Problem (SSP) is a well-known NP-complete problem which is defined as: given a set of positive integers $W\leftarrow \{p_1,\dots,p_n\}$ and a positive integer $s$ compute the set of boolean numbers $\{y_1,\dots,y_n\}$ such that  $p_1y_1+\dots+p_ny_n = s$.
 \begin{definition} (Subset-Sum Problem (SSP)).
 	\label{SSP}
 	Let $W\leftarrow \{p_1, p_2, \dots,p_n\}$ be a set of $n$ positive random integers. The SSP is the problem of finding a subset of $W$, where the sum of its members is equal to a given value $s$. 
 \end{definition}
Until now there exists no polynomial-time solution for the SSP, but there exist some pseudo-polynomial algorithms which can solve special instances of the SSP with small inputs. The complexity of the best possible solutions of the SSP are equal to $O(2^{n/2})$ \cite{KCSSP}, $O(s\sqrt{n})$ \cite{JVSSP} and $O(nc)$ \cite{P99}, where $c$ is the size of the largest element of $W$. In 2014, Cygan et al. showed that the SSP cannot be computed in time $s^\epsilon poly(n)$ under the Set Cover Hypothesis, where $\epsilon<1$ \cite{Cygan}. A variant of the SSP is the unbounded subset sum problem which is defined in Definition 1.2.
 \begin{definition} (Unbounded Subset-Sum Problem (USSP)).
 	\label{DSSP}
 	Consider a set $W\leftarrow \{p_1, p_2, \dots,p_n\}$ of $n$ random positive integers. Given a non-negative integer $s$, the USSP is the problem of finding the set of non-negative integers $\{y_1,\dots,y_n\}$ such that $s=\sum_{i=1}^{n}y_ip_i$, where $p_1 < p_2 < \dots <p_n$. Note that there is no condition on non-negative integers $y_i$'s (i.e. they can be zero or any positive integer). We use the notation $n-$USSP to explicitly indicate that the size of $W$ is $n$.
 \end{definition}
 The USSP is an NP-complete problem and the best known solution for it is a pseudo-polynomial-time algorithm with time complexity $O((p_{1})^{2}+n)$ and memory complexity $O(p_{1})$, where $p_{1}$ is the smallest element of $W$ \cite{PH}. This solution is exponential in term of the magnitude of the input and polynomial in the size of the input. 
 There is a reduction from the USSP to SSP, and the SSP problem is itself reducible to the partitioning problem \cite{reduce}, which is partitioning a given set of positive integers $W$ into two subsets $W_1$ and $W_2$ such that the sum of the elements of $W_1$ equals the sum of the elements in $W_2$. The reduction of USSP to SSP of \cite{reduce} is not correct since by transforming an instance of the SSP problem to an instance of the USSP problem the set of possible answers dramatically increases. Note that most of the answers of the transformed USSP problem do not work for the source SSP problem. If we want to compute the answer of the USSP problem which works for SSP too, we need to add some constraints. These constraints make the USSP problem the constraint satisfaction problem, which is NP-complete.\\
 The USSP can be considered as a special case of Unbounded Knapsack Problem (UKP), where the weight and the value of elements is equal. The only difference is that in USSP the amount of $\sum_{i=1}^{n}y_ip_i$ must be exactly equal to $s$, but in the UKP it can be equal to or less than $s$.
  Kellerer et al. described the UKP problem in Chapter Eight of Knapsack Problems \cite{UKP}.
  \subsection{Our Contribution}
  In this paper, we propose two versions of an algorithm for solving the USSP problem with various conditions on $s$. By using these two algorithms, we can solve many instances of USSP with large values of $n$ and $s$. 
  \begin{itemize} 
  	\item Algorithm 3 solves the USSP for $s >\sum_{i=1}^{k-1}q_iq_{i+1}-q_i-q_{i+1}$ with $O(n)$ computational complexity, where $q_i$'s are the elements of a small subset of $W$ that the $gcd$ of its elements divides $s$.
  	\item Algorithm 4 is a deterministic algorithm which can compute the answer of the USSP with probability between $0.5$ to $1$ for $p_1\le s \le \sum_{i=1}^{k-1}q_iq_{i+1}-q_i-q_{i+1}$. Its success probability is directly related to the number of subsets of $W$, in which $gcd$ of its elements divides $s$ (i.e. If the set $W$ contains a large number of prime integers, this probability will be close to one). Anyway, the USSP with small $n$ and $s$ does not have large space and can be solved by brute force algorithms. 
  \end{itemize}
  	\begin{table}[h]
  	\label{Table 1}
  	\caption{Brief Review of the Proposed Algorithms}
  	\begin{center}
  		\begin{tabular}{|p{65pt}|p{63pt}|p{40pt}|p{65pt}|}
  			\hline
  			\textbf{Algorithm}&\textbf{Complexity} & \textbf{$s$} &\textbf{Probability}    \\
  			\hline	
  			Algorithm 3&$O(n)$& $s>z_1$ $^*$  &  1\\
  			\hline	
  			Algorithm 4&$O(n^2)$& $s\le z_1$  &  $1-x$ $^{**}$\\
  			\hline
  \end{tabular} \\\setlength\tabcolsep{5pt}	
	\begin{flushleft} 	 $^{*}$ $z_1\leftarrow \sum_{i=1}^{k-1}q_iq_{i+1}-q_i-q_{i+1}$, where $q_i$'s are the elements of a small subset of $W$ in which $gcd$ of its elements divides $s$.\\
	$^{**}$ $x\leftarrow \prod_{i=1}^{i=f} (1-(\frac{z_{1}}{2z_{i}}))$, where $1\le f <n$ is the number of subsets of $W$ in which $gcd$ of their elements divides $s$ and Algorithm 4 can find them.
	 \end{flushleft} 
  	\end{center}
  \end{table}
\subsection{Paper Organization}
The rest of this paper is organized as follows. In Section \ref{rel}, the related work are discussed. In Section \ref{2ussp}, $2$-USSP problem is examined. An algorithm for the USSP with large values of $s$ is presented in Section \ref{n-ussp}. In Section \ref{small}, an algorithm for handling small values of $s$ is proposed. Finally, the paper is concluded in Section \ref{Conclution}.

 \section{Related Work}
 \label{rel}
 In 1996, Alfonsin proposed another variation of the SSP called Subset-sum with Repetitions Problem (SRP): given positive integers $s$, $r_1,\dots,r_n$ and an increasing set of coprime integers $W\leftarrow \{p_1,\dots,p_n\}$ there exist positive integers $y_i$, $y_i\le r_i$ such that $y_1p_1+\dots+y_np_n=s$ \cite{RA}. He stated that the SRP problem is NP-complete even if $W$ is a super-increasing set and $r_i\leftarrow 0, 1$ for all $i$. He also found a polynomial algorithm for the special case where the set $W\leftarrow \{p_1,\dots,p_n\}$ is a chain sequence. It means that $p_i|p_{i+1}$. Anyway their approach was not universal \cite{RA}.\\
 In 1996, Hansen and Ryan proposed an algorithm with time complexity $O((p_{1})^{2}+n)$ and memory complexity $O(p_{1})$ for the USSP problem \cite{PH}, where $W\leftarrow \{p_1,\dots,p_n\}$ are coprime integers. 
 In 2009, Muntean and Oltean proposed an optical solution to the problem of whether an instance of the USSP has an answer or not \cite{Optical}. 
 In 2017, Bringmann showed that the USSP problem can be solved in $O(s\log s)$ \cite{Near-linear}. 
 In 2018, Wojtczak showed that the unbounded subset-sum problem with rational numbers is strongly NP-complete. In other words, no pseudo-polynomial algorithm can exist for solving USSP with rational numbers unless P=NP \cite{ratinal}.\\
 Another variant of this problem is to decide whether an instance of the USSP has a solution or not. This is a YES/NO problem and it is also an NP-complete problem \cite{Optical}. More than one hundred years ago, Ferdinand Georg Frobenius introduced the Frobenius number as defined in Definition \ref{Frobenius number}.
\begin{definition} (Frobenius number).
	\label{Frobenius number}
	Let $W\leftarrow \{p_1,\dots,p_n\}$ be a set of coprime positive integers, where $p_1 < p_2 < \dots <p_n$. The greatest integer which cannot be expressed as a linear combination (with nonnegative integer coefficients) of elements of $W$ is called the Frobenius number $f(p_1,\dots,p_n)$. 
\end{definition}
Computing the Frobenius number for $n=2$ is an easy problem and it is equal to $f(p_1,p_2)=p_1p_2-p_1-p_2$ \cite{Sylveste}, \cite{UnsolvedProblem}. Furthermore, Sylvester showed that half of integers $0,1\dots, p_1p_2-p_1-p_2$ are not representable by any linear combination of $p_1$ and $p_2$ \cite{Sylveste}. In 1994, Davison proposed a quadratic-time algorithm for computing the Frobenius number for $n = 3$ \cite{Davison}.
The problem of finding the Frobenius number for $n>2$ is an NP-hard problem and the time complexity of the best solution for it is equal to $O(p_1\sqrt{n})$ \cite{Frobenius}. The upper bound of Frobenius number for $n>2$ is equal to $f(p_1,\dots,p_n)\le p_1(p_n-1)-p_n$ \cite{Frobenius}.

\section{Preliminaries}
\label{2ussp}
First we focus on the 2-USSP, which denotes the USSP with $n=2$ and $W\leftarrow \{p_1,p_2\}$. In Theorem \ref{YN2} we obtain the answers of the USSP equation $y_1p_1 + y_2p_2 = s$ for any adequately large random integer $s$.
\begin{theorem}  
	\label{YN2}
Let $p_1$ and $p_2$ be non-negative integers, where $p_1<p_2$ and $gcd(p_1,p_2)=1$, then the equation $y_1p_1+y_2p_2=s$ has $l\leftarrow \big\lfloor\frac{s-p_1(sp_{1}^{-1} \bmod p_2)}{p_1p_2}\big\rfloor+1$ answers $(y_1 , y_2)$. Moreover, if $s> p_1(p_2-1)-p_2$, then the equation $y_1p_1+y_2p_2=s$ has at least one answer $(y_1 , y_2)$.
\end{theorem}
\begin{proof}
	Let we rewrite the equation $y_1p_1+y_2p_2=s$ in modulo $p_2$. We have $p_1y_{1} \bmod p_2= s \bmod p_2$. Hence, we can compute an answer for $y_1$ as $y_{1}^{*}=sp_{1}^{-1} \bmod p_2$ (the largest possible value for $y_1^{*}$ is $p_2-1$). With this value for $y_1$ and any other integer value for $y_2$, the USSP equation $y_1p_1+y_2p_2=s$ is satisfied modulo $p_2$. 
	The value $y_{2}^{*}=\frac{s-p_1y_{1}^{*}}{p_2}$ omits the modulo if it is a non-negative integer. So, with this values for $y_1$ and $y_2$ the equation $y_1p_1+y_2p_2=s$ is satisfied (without modulo $p_2$). In fact if the equation $y_1p_1+y_2p_2=s$ has at least one answer, then $y_2^*=\frac{s-p_1y_1^*}{p_2}$ is an integer. Consequently, the only condition that must be held is that $y_{2}^{*}$ should be non-negative. The largest possible value of $y_1$ is equal to $p_2-1$, so the Frobenius number that is the possible largest value of $s$ for which the equation $y_1p_1+y_2p_2=s$ has no answer, is computed by considering $y_{2}^{*}=-1$ as below.
	\begin{align}
	\frac{s-p_1y_{1}^{*}}{p_2}&= -1 \nonumber \\
	\Rightarrow s&= p_1y_{1}^{*}-p_2 
	\end{align}	
	As the maximum possible value for $y_{1}^{*}$ is $p_2-1$, the maximum possible value for $s$ is obtained as below.
	\begin{equation}
	s= p_1 (p_2-1)-p_2 
	\end{equation}
	So the lower bound of $s$ is $p_1p_2-p_1-p_2$ in the sense that for any integer $s > p_1p_2-p_1-p_2$, the equation $p_1y_1+p_2y_2=s$ has at least one non-negative answer. If $s\le p_1p_2-p_1-p_2 $, then there may be no valid positive value for $y_2$ and the value of $y_2$ may be negative. \\

Now suppose $(y^{*}_1, y^{*}_2)$ is an answer for the USSP equation $y_1p_1+y_2p_2=s$. Then one can easily check that $(y^{*}_1+p_2, y^{*}_2-p_1)$ is also an answer for this equation if $y^{*}_{2}-p_1$ is still non-negative. This condition is equivalent to $s\ge (y^{*}_{1}+p_2)p_1$. As the extension we can say $(y_{1}^{*}+(l-1)p_2, y_{2}^{*}-(l-1)p_1)$ is also an answer to the USSP equation if $y^{*}_{2}-(l-1)p_1 \ge 0$.
Now suppose $y_{2}^{*}-lp_1<0$. Then, the value of $(y^{*}_1, y^{*}_2)$, as a seed, defines a coset of $l$ valid values for $y_1$ and $y_2$ as
	\begin{align}
	\{(y_{1}^{*}, y_{2}^{*}), (y_{1}^{*}+p_2, y_{2}^{*}-p_1), \dots, (y_{1}^{*}+(l-1)p_2, y_{2}^{*}-(l-1)p_1)\}	
	\end{align}  
	The value of $y_2$ must be non-negative, so the value of $l$ is obtained as below.\\
	\begin{align}
	&y^{*}_{2}-(l-1)p_1 \ge 0\\
	&\Rightarrow l-1=\bigg\lfloor\frac {y^{*}_{2}}{p_1} \bigg\rfloor
	=\bigg\lfloor \frac{s-p_1y_{1}^{*}}{p_1p_2}\bigg\rfloor\\
	&\Rightarrow l=\bigg\lfloor \frac{s-p_1y_{1}^{*}}{p_1p_2} \bigg\rfloor+1
	\end{align}
	 Note that if $gcd(p_1,p_2)= e$, then the equation $y_1p_1+y_2p_2=s$ has answer if $e|s$. If $gcd(p_1,p_2)=e\ne 1$ and $e|s$ then $p_1^{-1} \bmod p_2$ does not exist, so the value $sp_1^{-1} \bmod p_2$ cannot be computed. In such cases we can easily replace $p_1, p_2$ and $s$ by  $\frac{p_1}{e},\frac{p_2}{e}$ and $\frac{s}{e}$ in the $sp_1^{-1} \bmod p_2$. The answers will satisfy the original instance of the USSP. Note that if $gcd(p_1,p_2)=e\ne 1$ and $e\not|s$ then the equation does not have answer at all.
	\end{proof}
In 1884, Sylvester showed that half on integers $0\le s\le p_1p_2-p_1-p_2$ cannot be expressed by linear combination of $p_1$ and $p_2$ by non-negative coefficients.  
\begin{theorem}{(Sylvester Theorem.)}
Let $z=p_1p_2-p_1-p_2$. Exactly half of the integer values $s$ in the interval $0 \le s \le z$ can be represented by linear combination of $p_1$ and $p_2$ with non-negative coefficients. The following set include these numbers \cite{Sylveste}.
\begin{align} 
\{&0p_2+0p_1, 0p_2+p_1,\dots, 0p_2+\lfloor\frac{z}{p_1}\rfloor p_1,\nonumber \\
&p_2+0p_1, p_2+p_1,\dots, p_2+\lfloor\frac{z-p_2}{p_1}\rfloor p_1,\nonumber \\
&\rotatebox{90}{\dots} \nonumber \\ 
&\lfloor\frac{z}{p_2}\rfloor p_2+0p_1 \}
\end{align}
In other words, any non-negative integer $z-ip_1-jp_2$ for $i \in \{1,\dots,p_2-1\}$ and $j \in \{1,\dots,i\}$ is not representable by linear combination of $p_1$ and $p_2$ with non-negative coefficients. 
\end{theorem}
In Theorem \ref{frac}, we show that exactly $\lceil \frac{s}{p_1}\rceil+ \sum_{i=0}^{k}\lfloor\frac{s-ip_1}{p_2}\rfloor$ of integers between $0 \le r \le s$, $s\le p1p2-p1-p2$ can be expressed by linear combination of $p_1$ and $p_2$ with non-negative coefficients.
\begin{theorem}
\label{frac}
Let $z=p_1p_2-p_1-p_2$, then for any integer $s$ smaller than $z$, the $\lceil \frac{s}{p_1}\rceil+ \sum_{i=0}^{k}\lfloor\frac{s-ip_1}{p_2}\rfloor$ integers in the interval $0 \le r \le s$ can be expressed by linear combination of $p_1$ and $p_2$ with non-negative coefficients.
\end{theorem}
\begin{proof}
  Let's partition $\{0,\dots, z\}$ into $k=\lceil \frac{z}{p_1} \rceil$ sets $\{ip_1,\dots,(i+1)p_1-1\}$, for $i \in \{0,\dots,k\}$, (Ignore the lack of the last set.). One can easily check that the integers $r_1=ip_1$ and $r_2=jp_2$ will be  in $ip_1\le r_1< (i+1)p_1$ and $xp_1\le r_2< (x+1)p_1$ respectively, for $j \in \{0,\dots,\lfloor z/p_2\rfloor \}$, where $x= \lfloor \frac{jp_2}{p_1}\rfloor$. Let $e=jp_2 \bmod p_1$, then $e$'th element of any interval $\acute{i}\ge x+i$ can be expressed by $jp_2+ip_1$, for $i \in \{0,\dots,k\}$.
  As a result, for integers $ip_1\le r< (i+1)p_1$, the integer $ip_1$  plus another $\lfloor\frac{(i+1)p_1-1}{p_2}\rfloor$ integers can be expressed by $p_1$ and $p_2$.
  Let $s \in \{0,\dots,z\}$ then $\lceil \frac{s}{p_1}\rceil+ \sum_{i=0}^{k}\lfloor\frac{s-ip_1}{p_2}\rfloor$ integers $0\le r \le s$ can be represented by linear combination of $p_1$ and $p_2$.\\
 As observed, the number of integers which can be expressed by linear combination of $p_1$ and $p_2$ grows linearly with approaching $s$ to $z$. So, the percentage of $2-$USSP instances $y_1p_1+y_2p_2=r$, for $r =\{0,1,2,\dots,s\}$, $s\le z$ is approximated by $\frac{s}{2z}$. 
\end{proof}
  
  Based on Theorem \ref{YN2}, we define $A2(s,p_1,p_2)=y_1$ as an algorithm which computes the 2-USSP with time and memory complexity $O(1)$. The $A2(s,p_1,p_2)=y_{1}^{*}$ works as shown in Algorithm A2.\\
  \begin{algorithm}
  	\label{A2}
  	\textbf{Input: }$s, W=\{p_1,p_2\}$\\
  	\textbf{Output: }$y_{1}^{*}$\\
  	$e=gcd(p_1,p_2)$\\
	\eIf{$gcd(p_1,p_2)\not| s$}{\textbf{Return}($\bot$)}
	{$s\leftarrow s/e$\\
	$p_1\leftarrow p_1/e$\\
	$p_2\leftarrow p_2/e$\\
	$y_{1}^{*}= s p_{1}^{-1} \bmod p_2$\\
	\textbf{Return}$(y_{1}^{*})$\\}
	\caption{The algorithm $A2(s,p_1,p_2$)}
\end{algorithm} 
\section{Unbounded Subset-Sum Problem}
\label{n-ussp}
Assume $W=\{p_1,\dots , p_n\}$ such that for $j= \{1,\dots, n-1\}$ we have $p_{j+1} > p_j$, that is, $W$ is sorted in ascending order. The USSP can be solved just by invoking the $A2(s, p_{1}, p_{2})$ algorithm as follows:
\begin{align}
	s_1&=s\\
	y_1&=A2(s_1,p_1,p_2)\\
	s_2&=s_1-y_1p_1\\
	y_2&=A2(s_2,p_2,p_3)\\
	&\rotatebox{90}{\dots} \nonumber \\
	y_{n-1}&=A2(s_{n-1},p_{n-1},p_n)\\
	s_{n}&=s_{n-1}-y_{n-1}p_{n-1}\\
	y_n&=s_n/p_n
\end{align}
Some one may says there is no need for equations (11)-(14); it can stop at $s_2$ because $p_2|s_2$, so we have a solution as long as $s_2>0$. But its not true, because a correct solution is the one that picks from all elements of $W$ not just two of them. 
The algorithm $A2(s, p_{1},p_{2})$ works only if $gcd(p_{1},p_{2}$)$|s$ and $s>p_{1}(sp_1^{-1} \bmod p_{2})$. These conditions must be met just for the first pair ($p_1, p_2$). Because, if $gcd(p_1,p_2)|s$, then we have  $gcd(p_2,p_3)|p_2|s_2$ and as an extension we can say that $gcd(p_{i},p_{i+1})|s_{i}$. We have $n-1$ equations\\
\begin{align}
s&=s_2+y_1 p_1\\
s_2&=s_3+y_2 p_2\\
&\rotatebox{90}{\dots} \nonumber \\
s_{n-1}&=y_{n} p_{n}+y_{n-1} p_{n-1}
\end{align}
Then we have
\begin{align}
s+s_2+\dots+s_{n-1}&=y_1p_1+\dots+y_np_n+s_2+\dots+s_{n-1}\nonumber\\ 
\Rightarrow s&=y_1p_1+\dots+y_np_n
\end{align}   
As observed just by using A2 algorithm the USSP problem can be solved only if the greatest common divisor of the two smallest elements of $W$ divides $s$. To overcome this limitation, we propose an efficient algorithm denoted by Algorithm 3 with $O(n)$ computational complexity. This algorithm just needs $gcd(p_1,\dots, p_n)|s$. Let $W_1$ be a small subset of $W$, where $gcd$ of its elements divides $s$, then, in the first step, the Algorithm 2 finds $W_1$, and second, Algorithm 3 solves the USSP problem for $W_1$ and $s$, and in the final step, it solves the USSP problem for other elements of $W$. The Algorithm 2  and 3 works as follows.\\
\begin{algorithm}[H]
	\label{algorithm2}
	\textbf{Input: }$s, W=\{p_1,\dots,p_n\}$\\
	\textbf{Output: } A small subgroup of $W$, where $gcd$ of its elements divides $s$, or $\bot$\\
	\If{$gcd(p_1,\dots , p_n)\not|s$}{
		\textbf{Return}$(\bot)$}
	$q_1,\dots,q_n=0$\\
	$k=0$\\
	\For{$i=n \ to \ 1$}{
		\If{$gcd(q_1,\dots,q_k,W-\{p_i\}) \not| s$}{
			$k\leftarrow k+1$\\
			$q_{k}=p_i$ \color{blue} \ \% Element of  $W_1$.\color{black}\\	
		}
	$W \leftarrow W-\{p_i\}$ 		\color{blue} \ \% Delete $p_i$ from  $W$.\color{black}\\
	\If{$gcd(q_1,\dots,q_k)|s$}{$break$}
	}
	 \textbf{Return} $(W_1=\{q_k,\dots,q_1\})$ 	
	  \caption{Finding a small subset of $W$, where $gcd$ of its elements divides $s$.}
\end{algorithm}

The idea of Algorithm 3 is that we choose $y_i$ such that $gcd(q_2,\dots,q_k)$ divides $s-y_ip_i$, where $p_i=q_1$. As a result, the size of $W_1$ can be reduced to 2, and then, by running $A2(s,q_{k-1},q_k)$ the USSP can be solved. As observed, the Algorithm 3 works for any integer bigger than $q_{k-1}q_{k}-q_{k}-q_{k-1}+\sum_{i=1}^{k-1}q_ix_i-q_i-x_i$, where $x_i=gcd(q_i,\dots,q_k)$. In the worst case, $x_i=q_{i+1}$, so, in the worst case, the Algorithm 3 works for any integer bigger that $z_1=\sum_{i=1}^{k-1}q_iq_{i+1}-q_i-q_{i+1}$.\\

\begin{algorithm}[H]
	\label{algorithm3}
	\textbf{Input: }$s, W=\{p_1,\dots,p_n\}$\\
	\textbf{Output: }$(y_1,\dots ,y_n)$ such that $p_1y_1+\dots+p_ny_n=s$ or $\bot$\\
	$W_1$=$\{q_1,\dots,q_{k}\}$=Algorithm2($s, W=\{p_1,\dots,p_n\}$)\\
	$y_1,\dots,y_n=-1$\\
	\For{$i=1 \ to \ k-1$}{
		\For{$j=1 \ to \ n$}{  
			\If{$p_j=q_i$}{
				$break$ \color{blue} \ \% Find $q_i$ in $W$.\color{black}
			}
		}
		\eIf{$i<k-1$}
		{
			$x=gcd(q_{i+1},\dots ,q_{k})$
		}
		{
			$x=q_k$
		}
		\If{$s<p_j(sp_{j}^{-1} \bmod x)$}
			{Return($\bot$)\color{blue} \ \% $s$ is too small\color{black}}
		$y_j=A2(s, p_j,x)$ \color{blue} \ \% Solving the USSP for elements of $W_1$.\color{black}\\
		$s \leftarrow s-p_j y_j$ \color{blue} \% Now $gcd$ of $W_1-\{q_i\}$ divides $s$\\ 
	}
	\color{blue} \% Now $q_{k}$ divides $s$.\color{black}\\
	\For{$j=1 \ to \ n$}{
		\If{$p_j=q_{k}$}
		{$break$  \  \color{blue} \% Find $q_k$ in $W$.\color{black}}
	}
\end{algorithm}

\begin{algorithm}

	\For{$i=j+1 \ to \ n$}{
		\eIf{$s\ge p_j(sp_{j}^{-1} \bmod p_i)$}{
			$y_j=A2(s,p_j, p_{i})$\\
			$s \leftarrow s-p_jy_j$ \color{blue} \small \% Solving the problem for the rest of elements of $W$. \color{black}\\
			$j=i$	
	}

		{$y_i=0$}
	} \color{blue} \% Now $s$ is a multiple of $p_{j}$, where $j\le n$.\\ \color{black}
	       
	\For{$i=1 \ to \ j-1$}{
		\If{$y_i=-1$}{
			\eIf{$s\ge p_jp_i$}{
				$y_i=p_j$ \\ 
				$s \leftarrow s-p_iy_i$\color{blue} \ \% $s$ still is a multiple of $p_{j}$\\ \color{black}}
			{$y_i=0$}
	}}
	\eIf{$s=0$}{$y_j=0$}{$y_j=s/p_j$}
  \textbf{Return}$(y_1,\dots,y_{n})$ 
  	\caption{The proposed algorithm for solving the USSP}
 \end{algorithm}
\     
 \section{ An Algorithm for the USSP with Small $s$}
 \label{small}
 Let elements of $W$ be independent random integers uniformly chosen from $\{2,3,\dots,s\}$ and $s \in \{p_1,p_1+1,\dots,\sum_{i=1}^{k-1}q_iq_{i+1}-q_i-q_{i+1}\}$ be a uniform random integer, then the USSP problem can be solved with probability more than $0.5$ by Algorithm 3, if there is a subgroup $W_1=\{q_1,\dots,q_k\}$, in which the $gcd$ of its elements divides $s$. The Algorithm 4 can improve the success probability by calling Algorithm 3 with another subset $W_2$ (if there exists any), in which $gcd$ of its elements divides $s$. The success probability of running Algorithm 3 with two different subsets can be increased up to $0.5+0.5*0.5=0.75$. As observed, the success probability increases logarithmically, from $0.5$ to less than one, with increasing the number of subgroups of $W$ in which $gcd$ of their elements divides $s$. 
 The Algorithm 4 succeeds with high probability if the input is given at random and $n$ is a large integer (i.e. it does not work for all inputs).\\
 Note that, if the elements of $W$ are not independent then the number of subsets, in which $gcd$ of its elements divides $s$, decreases. As a result, the success probability of this algorithm decreases too, but it never falls less than $0.5$.
\begin{algorithm}[H]
	\label{algorithm4} 
	\textbf{Input: }$s, W=\{p_1,\dots,p_n\}$\\
	\textbf{Output: }$(y_1,\dots,y_n)$ or $\bot$\\
	\For{$i=1 \ to \ n$}{
		$W_1=q_1,\dots,q_k$=Algorithm2($s, W=\{p_1,\dots,p_n\}$)\\
		\eIf{Algorithm3($s, W_1=\{q_1,\dots,q_k\}$)=$\bot$}{
			\For{$j=1 \ to \ n$}{
				\If{$p_j=q_1$}{
				$break$ \color{blue} \ \% Find $q_1$ in $W$.\color{black}}
			}
			$W \leftarrow W-\{p_j\}$}
		{	$x_1,\dots,x_k$=Algorithm3($s, W_1=\{q_1,\dots,q_k\}$)\\
			$i=1$\\
				\For{$j=1 \ to \ n$}{
					\eIf{$q_i=p_j$}{$y_j=x_i$\\$i=i+1$}{$y_j=0$}
				}
				\textbf{Return}$(y_1,\dots,y_{n})$\\
			}
		}		
\textbf{Return}($(\bot)$)
	\caption{The proposed algorithm for solving USSP with small values of $s$}
\end{algorithm}
As observed, Algorithm 4 can find at most $n-1$ subsets of $W$ (if exists.), in which $gcd$ of their elements divides $s$. The success probability can be at most $1-{\frac{1}{2}}^{n-1}$, but the size of these subsets are different and therefore they have also different thresholds in Algorithm 3. So, to compute the success probability of Algorithm 4, we have to estimate the percentage of the integers $0\le s \le z_1$ which can be expressed by linear combination of elements of other $n-2$ subgroups.\\ 
 Let $W_2$ be another small subgroup of $W$ in which $gcd$ of its elements divides $s$, then, $50\%$ of integers $0\le s \le z_2$ can be represented by linear combination of elements of $W_2$, and based on Theorem \ref{frac}, about $\frac{z_{1}}{2z_{2}}$ of integers $\{0,\dots,z_{1}\}$ can be expressed by linear combination of elements of $W_2$. Let $1\le f <n$ be the number of subsets of $W$, in which $gcd$ of their elements divides $s$, then, we can estimate the number of integers $0\le s \le z_{1}$ which cannot be expressed by the linear combination of elements of $W$ as below.
\begin{align}
P_{Fail}=\prod_{i=1}^{i=f} (1-(\frac{z_{1}}{2z_{i}}))
\end{align}  
We can say that the Algorithm 4 works for any integer $0 \le s \le z_{1}$ with probability equal to 
\begin{align} 
P_{Success}&=1-\prod_{i=1}^{i=f} (1-(\frac{z_{1}}{2z_{i}}))
\end{align}
Note that, no matter how much the value of $n$ goes up, for the integers $0<s<p_1$ the USSP instance does not have answer. So, it is notable to consider $P_{Success}$ as the probability that the integers $p_1\le s\le z_{1}$ has an answer in USSP problem.
\begin{example}
Let $W=\{11, 13, 15, 19, 21\}$. So $p_1 = 11$ and $z_{1,2}=119$, then the Algorithm 4 can solve $85$ percent of USSP problem, where $p_1\le s \le z_1$, and almost all of this $15$ percentage lies in $p_1\le s \le \frac{z_1}{2}$. We show that our estimation is equal to $80$ percent.\\
\begin{align}
&\frac{z_{1}}{2z_{1}}=\frac{119}{238}=0.5\nonumber\\
&\frac{z_{1}}{2z_{2}}=\frac{119}{334}=0.35\nonumber\\
&\frac{z_{1}}{2z_{3}}=\frac{119}{418}=0.28\nonumber\\
&\frac{z_{1}}{2z_{4}}=\frac{119}{718}=0.16\nonumber\\
&P_{Success}=1-\prod_{i=1}^{i=f} (1-(\frac{z_{i}}{2z_{1}}))\nonumber \\
&=1-(1-0.5) \times (1-0.35) \times (1-0.28)\times (1-0.16) \nonumber \\
&\cong 0.80
		\end{align}
\end{example}

\section{Conclusion}
\label{Conclution}
The unbounded subset-sum problem is a well-known and well-studied NP-complete problem and there is a formal reduction from USSP to SSP. The USSP problem does not have any known polynomial solution. In this paper, we proposed two efficient polynomial-time algorithms for solving the special cases of the USSP problem with different conditions on $s$. The Algorithm 3 solves the USSP for $s >\sum_{i=1}^{k-1}q_iq_{i+1}-q_i-q_{i+1}$ with $O(n)$ computational complexity, where $q_i$'s are the elements of small subset of $W$ in which $gcd$ of its elements divides $s$. Algorithm 4 is dedicated for smaller values of $s$, where only several exceptions happen for some values of $s$ in the interval $p_1 \le s \le \sum_{i=1}^{k-1}q_iq_{i+1}-q_i-q_{i+1}$ in the sense that the USSP instance may have an answer but our algorithm fails to find it. We estimated the probability of this failure and noticed it decreases to zero by a small increase of $n$. The complexity of our algorithm is $O(n^2)$.

\end{document}